\documentclass[10pt,bezier]{article}
\usepackage{amsmath,amssymb,amsfonts,epsfig,epstopdf}
\usepackage[labelsep=period, skip=14pt]{caption}
\usepackage{indentfirst}

\newtheorem{prethm}{{\bf  Theorem}}

\newenvironment{thm}{\begin{prethm}{\hspace{-0.5
               em}{\bf .}}}{\end{prethm}}
\newtheorem{prepro}{{\bf  Theorem}}

\newtheorem{precor}{{\bf  Corollary}}

\newenvironment{cor}{\begin{precor}{\hspace{-0.5
               em}{\bf .}}}{\end{precor}}
\newtheorem{preconj}{{\bf  Conjecture}}

\newtheorem{preremark}{{\bf  Remark}}

\newenvironment{remark}{\begin{preremark}{\hspace{-0.5
               em}{\bf .}}}{\end{preremark}}
\newtheorem{prelem}{{\bf  Lemma}}

\newtheorem{preproof}{{\bf  Proof.}}

\newenvironment{proof}[1]{\begin{preproof}{\rm
               #1}\hfill{$\Box$}}{\end{preproof}}

\date{}
\title{Even and Odd Cycles Passing a Given Edge or a Vertex} \author{\bf\small\sc S. Akbari$^{a,c}$, K. Etemadi$^{b}$, P. Ezzati$^{b}$, M.
Ghadiri$^{b}$ \\
{\footnotesize {\em $^{\rm a}$Department of Mathematical Sciences,
Sharif University of Technology,
 Tehran, Iran}}\\ {\footnotesize {\em $^{\rm b}$Department of
 Computer Engineering, Sharif University of Technology,
 Tehran, Iran}}\\
{\footnotesize {\em
$^{\rm c}$School of Mathematics, Institute
for Research in Fundamental Sciences (IPM),}} \\
 {\footnotesize {\em P.O. Box $19395-5746$, Tehran, Iran}}}
\begin{document}
\maketitle
\begin{abstract} {In this paper we provide some sufficient conditions for the existence of an odd or even cycle that passing a given vertex or an edge in $2$-connected or $2$-edge connected graphs. We provide some similar conditions for the existence of an odd or even circuit that passing a given vertex or an edge in 2-edge connected graphs. We show that if $G$ is a $2$-connected $k$-regular graph, $k \geq 3$, then every edge of $G$ is contained in an even cycle. We also prove that in a $2$-edge connected graph, if a vertex has odd degree, then there is an even cycle containing this vertex.}\\

\end{abstract}
\section{Introduction}
{
Throughout this paper all graphs are simple with no loops and multiple edges. 
Let $G$ be a graph with vertex set and edge set $V(G)$ and $E(G)$,
respectively. If $v \in V(G)$, then $N(v)$ denotes the set of all neighbors of
$v$ and ${d}_{G}(v) = |N(v)|$ is called the \textit{degree} of $v$. If every vertex of $G$ has the same degree $k$, then $G$ is called a \textit{$k$-regular} graph.

A graph $G$ is said to be $k$-connected if it has more than $k$ vertices and
remains connected whenever fewer than $k$ vertices are removed, and is $k$-edge connected if it remains connected whenever
fewer than $k$ edges are removed.

A \textit{walk} is a sequence of vertices and edges ${v}_{0}, {e}_{1}, {v}_{1},
\ldots ,{e}_{k}, {v}_{k}$, such that for $i = 1, \ldots, k$, the edge ${e}_{i}$
has endpoints ${v}_{i-1}$ and ${v}_{i}$. A \textit{trial} is a walk with no
repeated edge. A \textit{circuit} is a trial that its endpoints are the same. Two paths are \textit{internally vertex disjoint} if they do not have any internal vertex in common. Let $C$ be a cycle in $G$ and $e = \{u, v\} \in E(G) \setminus
E(C)$. If $\{u, v\} \subseteq V(C)$, then $e$ is called a \textit{chord} of $C$.

There are some results on the existence of cycles passing a given subset of vertices or edges in a graph. In $1960$, Dirac proved that for every set of $k$ vertices in a $k$-connected graph there exists a cycle that passes through all vertices of the set, see  \cite{dir}. In $1977$, Woodall proved that given any $l$ disjoint edges in a $(2l-2)$-connected graph, $l \geq 2$, there is a cycle containing all of them, see \cite{woo}. In $1981$, Bondy and Lov\`{a}sz showed that if $S$ is a set of $k$ vertices in a $k$-connected graph $G$, $k \geq 3$, then there exists an even cycle in $G$ through every vertex of $S$, see \cite{bon}. In \cite{hag}, H\"{a}ggkvist and Thomassen showed that if $L$ is a set of $k$ independent edges in a graph $G$ such that any two vertices incident with $L$ are connected by $k+1$ internally disjoint paths, then $G$ has a cycle containing all edges of $L$. In this paper, we consider some conditions on $2$-connected graphs eventuating existence of cycles with different parity that passes a given vertex or a given edge.

Here, we prove that if $G$ is a $2$-connected $k$-regular graph, $k
\geq 3$,
then every edge of $G$ is contained in an even cycle. We also show that if $G$ is a $2$-edge
connected graph and $v$ is a vertex of odd degree in $G$, then $v$ is contained in an even cycle. Also, we prove that if $G$ is a $2$-connected non-bipartite graph, then every edge of $G$ is contained in an odd cycle. Finally, we show that if $G$ is a $2$-edge connected $k$-regular graph, $k \geq 3$, then every edge of $G$ is contained in an even circuit.
}

\section{Cycles in Graphs Passing a Given Vertex or an Edge} \label{even}
{We start this section with the following theorem.}

\begin{thm}\label{thm0} {Let $G$ be a $2$-connected graph and $C$ be an odd
cycle in G. Then every $e \in E(G) \setminus E(C)$ is contained in an even cycle.}
\end{thm}
\begin{proof}{
We claim that there are two vertex disjoint paths starting from two end points of $e = \{u, v\}$ to two
end points of an arbitrary edge $f \in E(C)$. We add a new vertex on both $e$ and $f$. Clearly,
$G$ remains $2$-connected. Therefore, there are two internally vertex disjoint paths
between these two new vertices \cite[p.161]{wes}. These paths without new edges are the desired paths and the claim is proved.
Call these two paths $P$ and $Q$. Suppose that $P=u {u}_{2} \cdots
{u}_{n}$ and $Q=v {v}_{2} \cdots {v}_{m}$. Let ${u}_{i}$ and ${v}_{j}$ be the first vertices of $P$ and $Q$ in $V(P) \cap
V(C)$ and $V(Q) \cap V(C)$, respectively. Consider two paths $P' =u {u}_{2} \cdots {u}_{i}$ and $Q' =v {v}_{2}
\cdots {v}_{j}$. Now, we have
three cases:

{\bf Case 1.} $V(C) \cap \{u, v\} = \varnothing$. Suppose
that $P$ and $Q$ are two paths from ${u}_{i}$ to ${v}_{j}$ such that $V(P) \cup V(Q) = V(C)$ and $V(P) \cap V(Q) =
\{{u}_{i}, {v}_{j}\}$. Since $C$ is an odd cycle the parity of $P$ and $Q$ are different.
Hence one of the cycles $eP'PQ'$ and $eP'QQ'$ is an even cycle, as desired.

{\bf Case 2.} $V(C) \cap \{u, v\} = \{v\}$. Since $G$ is $2$-connected there exists a shortest path from $u$ to $C$
not containing $v$. We call this path by $S$ and let $V(C) \cap V(S)=\{s\}$.
Suppose that $P$ and $Q$ are two paths from $s$ to $v$, $V(P) \cup V(Q) = V(C)$ and
$V(P) \cap V(Q) = \{s,v\}$. Since $C$ is an odd cycle the parity of $P$ and $Q$
are different.
Hence one of the cycles $eSP$ and $eSQ$ is even.

{\bf Case 3.} The edge $e$ is a chord of $C$. Suppose that $P$ and $Q$ are two paths
from $u$ to $v$ such that $V(P) \cup V(Q) = V(C)$ and $V(P) \cap V(Q) = \{u,v\}$. Since
$C$ is an odd cycle so one of the cycles $eP$ and $eQ$ is even. The proof is
complete.
}
\end{proof}

{Now, we have the following corollary.}

\begin{cor} \label {corthm0}{
Let $G$ be a $2$-connected graph. If removing of every edge of $G$ does not make the graph bipartite, then every edge of $G$ is contained in an even cycle.
}
\end{cor}

\begin{thm}\label{cor1}{ Let $G$ be a $2$-connected non-bipartite graph. Then
every edge of $G$ is contained in an odd cycle.}
\end{thm}
\begin{proof}{ Since $G$ is non-bipartite so it has at least one odd cycle. Let
$C$ be an odd cycle and $e \in E(G)$. If $e \in E(C)$, then we are done. If $e \notin E(C)$, then the proof is similar to Theorem
$\ref{thm0}$ and $e$ is contained in an odd cycle. The proof is complete.}
\end{proof}

\begin{remark}{\rm{ The $2$-connectivity condition in Theorem $\ref{cor1}$ is not superfluous, as shown in Figure $1$. The graph in Figure $1$ is non-bipartite but this graph is not $2$-connected. The edge $e$ is not contained in an odd cycle.}}
\end{remark}

\begin{figure} [!htb]

\centering
\includegraphics{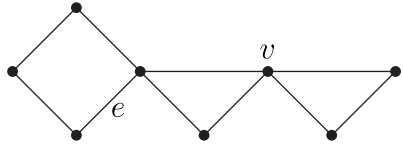}
\caption{The edge $e$ is not contained in an odd cycle. The vertex $v$ is not contained in an even cycle.}
\label{fig:graph1}
\end{figure}

{By Theorem \ref{thm0}, we have the following result.}

\begin{thm}\label{thm1}{ Let $G$ be a $2$-connected graph and $k \geq
3$ be a positive integer. If all vertices of $G$ have degree divisible by $k$, then every edge of $G$ is contained in an even cycle.}
\end{thm}
\begin{proof}{ First assume that $G$ is bipartite. Since $G$ is
$2$-connected so every edge of $G$ is contained in an even cycle \cite[p.162]{wes}. Next,
suppose that $G$ contains an odd cycle, say $C$. Let $e = \{u, v\} \in E(G)$. If there exists an odd cycle not containing $e$, then by Theorem $\ref{thm0}$, $e$ is contained in an even cycle. Thus assume that every odd cycle of $G$ contains
$e$.
Obviously, $H = G \setminus e$ is a bipartite graph. Since $e \in E(C)$, there is a path of even length between $u$ and
$v$ in $H = (X,Y)$. Clearly, $u$ and $v$ are in the same part of $H$, say $X$. Since
${d}_{H}(u) \equiv {d}_{H}(v) \equiv -1 \pmod{k}$, so the sum of vertex degrees in $X$ is $-2$ $\pmod{k}$, but the sum of vertex degrees in $Y$ is $0$ $\pmod{k}$, a
contradiction. 
The proof is complete.}
\end{proof}

\begin{cor} \label {corthm1}{Let $G$ be a $2$-connected $k$-regular graph, $k \geq
3$. Then every edge of $G$ is contained in an even cycle.}
\end{cor}

\begin{remark}{\rm{ The divisibility condition in Theorem $\ref{thm1}$ is required.
To see this, look at the Figure $2$. The edge $e$ of the following $2$-connected graph is not contained in an even cycle.}}
\end{remark}

\begin{figure} [!htb]

\centering
\includegraphics{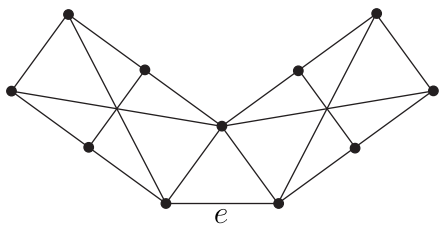}
\caption{The edge $e$ is not contained in an even cycle.}
\label{fig:graph1}
\end{figure}

\begin{remark}{\rm{ The $2$-connectivity condition in Theorem $\ref{thm1}$ is not superfluous.
To see this, look at the Figure $3$. In this figure an edge between a component and a vertex
means that all vertices of the component are adjacent to the vertex. Also an edge
between two components means all vertices of these components are adjacent. The graph in Figure $3$ is $k$-regular but this graph is not $2$-connected. Obviously, the edge $e$ is not contained in an even cycle.}}
\end{remark}

\begin{figure}[!htb]
\centering
\includegraphics{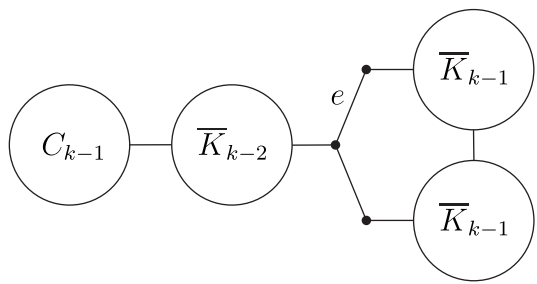}
\caption{A connected graph containing an edge which is not contained in an even cycle. 
}
\label{fig:graph2}
\end{figure}

\begin{thm}\label{thm2}{ Let $G$ be a $2$-connected graph, $v \in V(G)$ and $d(v) \geq 3$. Then $v$ is contained in an even cycle. }
\end{thm}
\begin{proof}{ Since $G$ is $2$-connected,
$v$ is contained in a cycle. Call this cycle by $C$. If $C$ is an even cycle, then we
are done.
Thus assume that $C$ is an odd cycle. Since $d(v) \geq 3$, $v$ is incident
with another edge $e$, not contained in $C$. Since $e
\notin E(C)$, so by Theorem $\ref{thm0}$, $e$ is contained in an even cycle.}
\end{proof}

\begin{remark}{\rm{ The $2$-connectivity condition in Theorem $\ref{thm2}$ is not superfluous.
To see this, look at the Figure $1$.}}
\end{remark}

\begin{thm}\label{thm3}{ Let $G$ be a $2$-edge connected graph and $v \in V(G)$.
If $d(v)$ is odd, then there is an even cycle containing $v$.}
\end{thm}

\begin{proof}{We claim that $G \setminus v$ has a connected
component $H$ such that $|V(H) \cap N(v)| \geq 3$. Since $G$
is $2$-edge connected, in each connected component of $G \setminus v$ there are
at least two neighbors of $v$. Now, assume that for each connected component
of $G$, say $H$, $|V(H) \cap N(v)| = 2$. This implies that $d(v)$ is even, a
contradiction. Thus, there exists a connected component of $G \setminus v$, say
$H$, such that $|V(H) \cap N(v)| \geq 3$.

Let $\{x,y,z\} \subseteq V(H) \cap N(v)$.
With no loss of generality assume that $\{x, y\}$ has minimum distance among all pairs of $\{x, y, z\}$ in $G \setminus v$. Suppose that $M$ is a path of minimum length
between $x$ and $y$. Since $H$ is connected, there exists a shortest path, say $N$, between $z$ and $M$. Let $w \in V(N)
\cap V(M)$ and $P$ and $Q$ are two paths such that $V(P) \cup V(Q) = V(M)$, $V(P) \cap
V(Q) = \{w\}$, $x \in V(P)$ and $y \in V(Q)$. We denote the edges $vx$, $vy$ and $vz$ by ${e}_{vx}$, ${e}_{vy}$ and ${e}_{vz}$, respectively. If two cycles ${e}_{vx}M{e}_{vy}$
and ${e}_{vy}QN{e}_{vz}$ are odd, then $l({e}_{vx}PN{e}_{vz}) = l({e}_{vx}M{e}_{vy}) + l({e}_{vy}QN{e}_{vz}) -
2l({e}_{vy}Q)$ which is even, where $l(R)$ denotes the length of $R$, as desired.}
\end{proof}

\begin{remark}{\rm{ Being odd for $d(v)$ is required in Theorem $\ref{thm3}$.
The graph in Figure $4$ is $2$-edge connected but the degree of $v$ is even and $v$ is contained in no even cycle.}}
\end{remark}

\begin{figure}[!htb]
\centering
\includegraphics{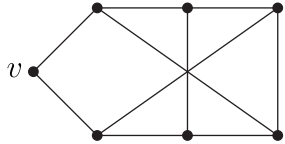}
\caption{The vertex $v$ in the graph has even degree and there is no even cycle containing $v$.}
\label{fig:graph3}
\end{figure}

{Now, we present some results on the existence of even or odd circuit passing a given edge or vertex.}

\begin{thm}\label{thm00}{\ Let $G$ be a $2$-edge connected graph and $C$ be an
odd cycle in $G$. Then every $e \in E(G) \setminus E(C)$ is contained in an even
circuit.}
\end{thm}
\begin{proof} {The proof is similar to Theorem $\ref{thm0}$.}
\end{proof}

{Using the method applied in the proof of Theorem \ref{thm1}, one can prove the following result.}

\begin{thm}\label{thm5}{\ Let $G$ be a $2$-edge connected graph and $k \geq
3$ be a positive integer. If all vertices of $G$ have degree divisible by $k$, then every edge of $G$ is contained in an even circuit.}
\end{thm}

\begin{remark}{\rm{ The divisibility condition in Theorem $\ref{thm5}$ is
required. To see this, look at the Figure 2.
}}
\end{remark}

{If one applies the idea of the proof of Theorems \ref{cor1} and \ref{thm2}, then the following results hold.}

\begin{thm}\label{thm6} {Let $G$ be a $2$-edge connected graph and $v \in V(G)$.
 If $d(v) \geq 3$, then $v$ is contained in an even circuit. }
\end{thm}

{We close this paper with the following result.}

\begin{thm}\label{thm7} {Let $G$ be a $2$-edge connected non-bipartite graph.
Then every edge of $G$ is contained in an odd circuit.}
\end{thm}

\noindent {\small {\sc   Saieed Akbari \quad       {\tt
s\_akbari@sharif.edu}}}\\
\noindent {\small {\sc  Khashayar Etemadi \quad       {\tt
etemadi@ce.sharif.edu}}}\\
\noindent {\small {\sc  Peyman Ezzati \quad       {\tt
pezzati@ce.sharif.edu}}}\\
\noindent {\small {\sc  Mehrdad Ghadiri \quad       {\tt
ghadiri@ce.sharif.edu}}}\\
\end{document}